\documentclass[preprintnumbers,floatfix,superscriptaddress,pra,twocolumn,showpacs,notitlepage,longbibliography]{revtex4-1}
\usepackage[utf8]{inputenc}
\usepackage[T1]{fontenc}
\usepackage[sc,osf]{mathpazo}
\usepackage{amsmath, amsthm, amssymb,amsfonts,mathbbol,amstext}
\usepackage{graphicx}
\usepackage{dcolumn}
\usepackage{bm}
\usepackage{bbm}
\usepackage{hyperref}
\usepackage{mathtools}
\usepackage{comment}
\usepackage{color}
\usepackage{multirow}

\def\RR{\mathbbm{R}}

\def\1{\mathbf{1}}
\def\0{\mathbf{0}}

\def\minimize{\textrm{minimize}}
\def\st{\textrm{subject to }}

\newcommand{\mean}[1]{\left\langle #1 \right\rangle}

\newcommand{\eg}{{\it{e.g.~}}}

\newtheorem{prop}{Proposition}

\newtheorem{theorem}[prop]{Theorem}
\newtheorem{definition}[prop]{Definition}

\newtheorem{lemma}[prop]{Lemma}

\renewcommand{\rho}{\varrho}

\newcommand{\processnext}[1]{%
  \ifx\listfinish#1\empty\else\listact{#1}\expandafter\processnext\fi}

\newcommand{\figref}[1]{Fig.~\ref{#1}}

\newcommand{\rafael}[1]{{ #1}}
\newcommand{\dani}[1]{{ #1}}
\newcommand{\lo}[1]{{ #1}}

{\begin{framed}\begin{small}}
{\end{small}\end{framed}}

\makeindex

\begin{document}
\title{{Causal h}ierarchy of multipartite Bell nonlocality}
\date{\today}

\author{R. Chaves}
\affiliation{International Institute of Physics, Federal University of Rio Grande do Norte, 59078-970, P. O. Box 1613, Natal, Brazil}
\affiliation{Institute for Theoretical Physics, University of Cologne, 50937 Cologne, Germany}
\author{D. Cavalcanti}
\affiliation{ICFO-Institut de Ciencies Fotoniques, The Barcelona Institute of Science and Technology, 08860 Castelldefels (Barcelona), Spain}
\author{L. Aolita}
\affiliation{Instituto de F\'isica, Universidade Federal do Rio de Janeiro, P. O. Box 68528, Rio de Janeiro, RJ 21941-972, Brazil}

\begin{abstract}
As with entanglement, different forms of Bell nonlocality arise in the multipartite scenario. These can be defined in terms of relaxations of the causal assumptions in local hidden-variable theories. However, a characterisation of all the forms of multipartite nonlocality has until now been out of reach, mainly due to the complexity of generic multipartite causal models. Here, we employ the formalism of Bayesian networks to reveal connections among different causal structures that make a both practical and physically meaningful classification possible. Our framework holds for arbitrarily many parties. We apply it to study the tripartite scenario in detail, where we fully characterize all the nonlocality classes. Remarkably, we identify new highly nonlocal causal structures that cannot reproduce all quantum correlations. This shows, to our knowledge, the strongest form of quantum multipartite nonlocality known to date. Finally, as a by-product result, we derive a non-trivial Bell-type inequality with no quantum violation. Our findings constitute a significant step forward in the understanding of multipartite Bell nonlocality and open several venues for future research.
\end{abstract}

\maketitle

\section{Introduction}
Bell nonlocality \cite{Bell1964,BCP+14} represents a fundamental and intriguing aspect of nature: Not all correlations observed in space-like separated measurements can be explained by any classical model respecting the causal assumptions of locality -- that the measurement outcomes depend only on local variables -- and measurement independence -- that the observers choose their measurement settings freely. Thus, understanding to what extent one has to give up the causal assumptions for a classical model to reproduce nonlocal correlations provides both insights into the nature of quantum correlations \cite{Spekkens2015,Ringbauer2016} and a natural way of quantifying them \cite{Hall2010,Hall2011,Putz2014,Chaves2015b,Chaves2016}. In view of that, the study of causal relaxations in Bell scenarios has become a topic of intense interest \cite{Hall2010,Hall2011,Putz2014,Chaves2015b,Chaves2016,Brans1988,Brassard99,Gisin99,Toner03,Pironio2003,Regev09,Barrett2010,Maxwell2014,brask2016bell}. In particular, in the bipartite scenario, it is known that all quantum correlations can be reproduced if either of the two causal assumptions are relaxed (see e.g. \cite{Hall2011}). However, this is not in general the case in the multipartite scenario.

In the $N$-partite case, even if $N-1$ parties communicate -- a relaxation of locality --, classical models cannot explain all quantum correlations \cite{Svetlichny1987,Collins2002,Seevinck2002}. This discovery gave rise to the notion of genuinely multipartite nonlocality (GMNL), with both fundamental and applied implications \cite{BGLP11,Moroder2013,Liang2014}. Since then, several forms of GMNL have been identified \cite{Jones2005,Bancal2009,ACSA10,Bancal2011,Gallego2012,Aolita12,Bancal2013,Saha2015}. However, a unifying picture of the models arising from all the different causal relaxations, together with the classes of nonlocality they lead to, has until now been an unrealistic goal. A significant obstacle is the rapidly increasing complexity of generic multipartite causal structures as $N$ grows.

Here, we develop a systematic characterisation of the \emph{classes of multipartite Bell nonlocality} in terms of the causal relaxations required for a classical causal model to explain all the correlations in the class. We use the formalism of Bayesian networks \cite{Pearl2009,Spirtes2001}, which allows us to identify equivalences among different causal relaxations in the context of nonsignaling correlations
. As a result, one can define \emph{causal classes of Bell correlations} each of which groups together many different causal structures.
This enormously simplifies the problem, rendering a practical characterisation possible. Additionally, the classification has the built-in advantage of automatically yielding a natural hierarchy, from which one can directly read which nonlocality classes are stronger than others in the inclusion sense.

We develop the formalism in full generality and discuss in detail the tripartite scenario, where a full characterisation of Bell nonlocality is given. The hierarchy delivers 10 classes of tripartite causal models with nonsignaling violations.
Interestingly, we prove that, in the nonsignaling scenario, 3 classes in different levels of the hierarchy collapse to a same single class. This leaves us with 8 inequivalent causal Bell classes. From these, at least 7 are violated by quantum correlations, including 1 class for which no quantum violations were known to date \cite{Jones2005}. This proves that nature is nonlocal in a stronger sense than previously known, closing a long-standing open question  \cite{Jones2005}.
Interestingly, the causal class for which we could not find a quantum violation produces correlations able to win, with unit probability, the celebrated nonlocal game without a quantum advantage "guess your neighbour's input" (GYNI) \cite{Almeida2010}. We identify a non-trivial Bell-type inequality with no quantum violation for this class, which can also be interesting on its own \cite{Winter2010}.

\section{Scenario}
We consider the correlations arising when $N$ parties perform local measurements on their respective shares of a joint $N$-partite system. These correlations are described by a conditional probability distribution $p(a_1,\dots,a_N \vert x_1, \dots, x_N)$,
where $x_i$ and $a_i$ label, respectively, the measurement choices (inputs) and outcomes (outputs) of the $i$-th party, for $i=1, \hdots , N$. If the parties are space-like separated, the correlations must be nonsignaling. That is, the marginal conditional probabilities $p(a_1,\dots, a_{i-1},a_{i+1},\dots,a_N \vert x_1,\dots, x_N)\coloneqq\sum_{a_i}\, p(a_1,\dots, a_i,\dots, a_N \vert x_1,\dots,x_N)$ over all but the $i$-th outcome must be independent on the $i$-th input \cite{BCP+14}:
\begin{align}
\label{eq:no_signalling}
\nonumber
&p(a_1,\dots, a_{i-1},a_{i+1},\dots,a_N \vert x_1,\dots, x_N)=\\
&p(a_1,\dots, a_{i-1},a_{i+1},\dots,a_N \vert x_1, \dots, x_{i-1}, x_{i+1},\dots, x_N),
\end{align}
for all $a_1,\dots, a_{i-1},a_{i+1},\dots,a_N$,  all $x_1,\dots, x_N$, and all $i$.
This means that the input choice by the $i$-th party cannot influence the statistics observed by the others. Our goal is to study the nonsignaling correlations that can arise from different causal structures where relaxations of measurement independence  and locality are allowed.

Causal structures can be represented with \emph{directed acyclic graphs} (DAGs) \cite{Pearl2009}, examples of which are shown in Figs. \lo{\ref{fig:equiv}, \ref{fig:belleq}, and \ref{fig:hierarchy}}. Each node of a DAG represents a classical random variable, and each directed edge encodes a causal relation between two nodes. For each edge, one calls the start vertex the \emph{parent} and the arrival one the \emph{child}. The acyclicity of the graph prevents an effect from being its own cause. Then, given a collection of variables $\mathcal{V}=\{v_1, v_2, \dots,v_n\}$, $\mathcal{V}$ forms a \emph{Bayesian network} with respect to a DAG $\mathcal{G}$ if the joint probability distribution $p(v_1,\dots,v_n)$ describing the statistics of $\mathcal{V}$ can be decomposed as
\begin{equation}
\label{markov}
p(v_1,\dots,v_n)= \prod^{n}_{i} p(v_i \vert {\rm pa}(v_i)),
\end{equation}
where ${\rm pa}(v_i)$ denotes the set of parents of $v_i$ according to $\mathcal{G}$. Here, we are interested in a specific subclass of structures with two common features.
First, they all possess an unobservable node, the \emph{hidden variable} $\lambda$, and two sets of observable ones, the inputs $x_1, \dots, x_N$ and the outputs  $a_1,\dots,a_N$, i.e. $\mathcal{V}_\text{Bell}=\{\lambda, a_1,\dots,a_N, x_1, \dots, x_N\}$. Second, each $i$-th output $a_i$ contains $x_i$ and $\lambda$ as parents, i.e. ${\rm pa}(a_i)\supseteq \{x_i,\lambda\}$.
We refer to these DAGs as \emph{Bell DAGs} (BDAGs).

The simplest BDAG, shown at the top of Figs. \ref{fig:belleq} and \ref{fig:hierarchy} for $N=3$, is the one for which ${\rm pa}(a_i) = \{x_i,\lambda\}$ for all $i$. This gives rise to the so-called \emph{local hidden-variable} (LHV) models, with correlations of the form \cite{Bell1964}:
\begin{align}
\label{p_classical}
\nonumber
&p(a_1,\dots,a_N \vert x_1, \dots, x_N)= \\
&\sum_{\lambda} p(\lambda)p(a_1 \vert x_1,\lambda) \hdots p(a_N \vert x_N,\lambda).
\end{align}
More complex causal structures are obtained by considering causal relaxations of the LHV BDAG. Relaxations of measurement independence \cite{Hall2010,Barrett2010,Hall2011} and locality \cite{Brassard99, Gisin99,Toner03,Regev09} have been studied in the bipartite scenario. In the multipartite case, communication among $N-1$ out of $N$ parties is allowed in \emph{bi-LHV} models \cite{Svetlichny1987,Collins2002,Seevinck2002,Bancal2009}, while causal influences from the input of one party towards the outputs of all others is accounted for in \emph{input-broadcasting} models \cite{Bancal2009}.
However, a systematic classification of $N$-partite causal structures for all the different causal relaxations was missing \footnote{Both bi-LHV and input-broadcasting models are particular cases of a generic family called \emph{partially paired models} in Ref. \cite{Jones2005} (see Fig. \ref{fig:hierarchy}). An input-communication structure is called \emph{totally paired} if every pair of inputs is sent to some output; otherwise it is partially paired. Partially paired models were shown to satisfy the Svetlichny inequality, which admits quantum violations. In contrast, it was an open question whether there exist totally paired models with quantum violations \cite{Jones2005}.}. The problem consists of organising, in a physically meaningful and practical way, objects with an exponential complexity in the number of parties. In what follows, we propose a solution to this problem and, on the way, give a positive answer to the above-mentioned open question.

\section{Bell classes of multipartite causal networks}
 Our classification relies on two technical results that connect different causal relaxations in the nonsignaling framework.
We say that a BDAG $\mathcal{G}_1$ \emph{nonsignaling implies} another $\mathcal{G}_2$, if every nonsignaling correlations compatible with $\mathcal{G}_1$ (i.e., produced by a Bayesian network with respect to it) are also compatible with $\mathcal{G}_2$. Note that, if all the causal relaxations in $\mathcal{G}_1$ are also present in $\mathcal{G}_2$, $\mathcal{G}_1$ automatically nonsignaling implies $\mathcal{G}_2$.
In addition, if $\mathcal{G}_1$ and $\mathcal{G}_2$ nonsignaling imply one another, we say that they are \emph{nonsignaling equivalent}. In particular, notice that if two BDAGs that are nonsignaling equivalent, the maximal violation of any given Bell inequality is the same and the correlations produced are useful for the same information-theoretic protocols.

The first result asserts that the most general causal influence from one party to another is nonsignaling equivalent to a single locality relaxation from the input of the former towards the output of the latter. This is depicted in Fig. \ref{fig:equiv} a) for $N=2$. Thus, any particular locality relaxation is nonsignaling implied by the input-to-output one.
\begin{lemma}[Generic locality relaxation $\leftrightarrow$ input-to-output locality relaxation]
\label{lemma:input-to-output}
Let $\mathcal{G}_{\rm gen}$ and $\mathcal{G}_{\rm in-out}$ be any two BDAGs whose only difference is that there exists $1\le j\neq i\le N$ such that, for $\mathcal{G}_{\rm gen}$, ${\rm pa}(x_j)\supseteq \{a_i,x_i\}$ and ${\rm pa}(a_j)\supseteq \{a_i,x_i\}$, whereas, for $\mathcal{G}_{\rm in-out}$, ${\rm pa}(a_j)\supseteq \{x_i\}$. Then, $\mathcal{G}_{\rm gen}$ and $\mathcal{G}_{\rm in-out}$ are nonsignaling equivalent.
\end{lemma}

\begin{figure}[t!]
\center
\includegraphics[width=0.8\columnwidth]{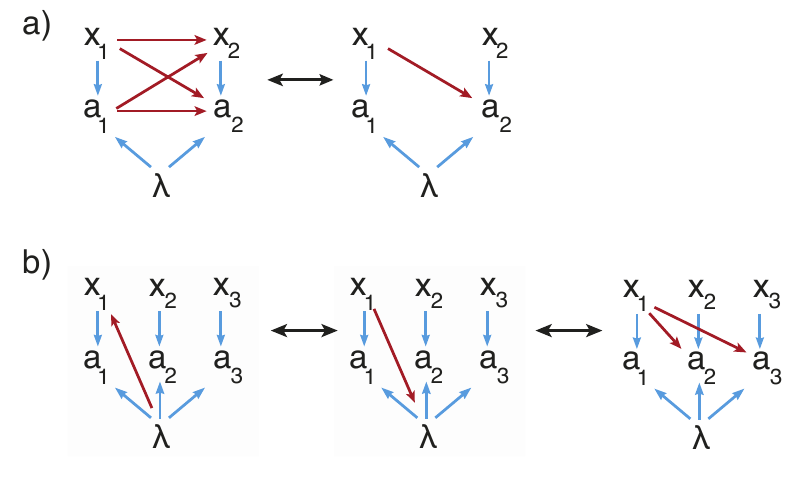}
\caption{a) The set of nonsignaling correlations produced by the most general locality relaxation from one party to another in the left-hand side coincides with the one produced by the input-to-output locality relaxation in the right-hand one. b) The relaxations of measurement independence in the left-hand side and the centre produce the same set of nonsignaling correlations as the input-broadcasting model in the right-hand side.}
\label{fig:equiv}
\end{figure}

\begin{proof}
We prove this lemma explicitly for the particular case of BDAGs of $N=2$ parties. The proof for the general case is totally analogous.
We have to prove the implication relations between the DAGs in Fig. 2a). The most general locality relaxation between two parties is represented by the BDAG $\mathcal{G}_{\rm gen}$ in the left-hand pannel of Fig 2a). The simpler DAG $\mathcal{G}_{\rm in-out}$ is represented in the right-hand pannel of Fig. 2a).
Clearly, since all the causal relaxations in $\mathcal{G}_{\rm in-out}$ belong also to the set of causal relaxations in $\mathcal{G}_{\rm gen}$, $\mathcal{G}_{\rm in-out}$ nonsignalling implies $\mathcal{G}_{\rm gen}$. We now prove that the converse also holds true, thus proving that both DAGs are nonsignaling equivalent.

Any probability distribution compatible with $\mathcal{G}_{\rm gen}$ can be decomposed as
\begin{align}
\label{wiring1}
\nonumber
& p(a_1,a_2 \vert x_1,x_2) \\ & = \sum_{\lambda} p(a_1,a_2,\lambda \vert x_1,x_2) \\
&= \sum_{\lambda} p(a_1, a_2 \vert x_1, x_2,\lambda) p(\lambda \vert x_1,x_2)
\\
\label{eq:cond_independence_app0}
&= \sum_{\lambda} p(a_1, a_2 \vert x_1, x_2,\lambda) p(\lambda \vert x_2) \\
\label{eq:cond_independence_app}
&  = \sum_{\lambda}  p(a_1 \vert x_1,\lambda)p(a_2 \vert a_1,x_1,x_2,\lambda) \frac{p(x_2\vert\lambda)}{p(x_2)} p(\lambda)\\
\label{eq:det_func}
&  = \sum_{\lambda}  p(a_1 \vert x_1,\lambda)p( a_2 \vert x_1, x_2,\lambda)  \frac{p(x_2\vert\lambda)}{p(x_2)} p(\lambda),
\end{align}
where Bayes' rule has been repeatedly used. In addition, Eq. \eqref{eq:cond_independence_app0} follows, using  Eq. (2),  from the fact that, for the BDAG  $\mathcal{G}_{\rm gen}$, $x_1$ is neither a descendant nor a parent of $\lambda$. In turn, Eqs. \eqref{eq:cond_independence_app} and  \eqref{eq:det_func} use the fact that, if the hidden variable $\lambda$ can take sufficiently many values (which can, without loss of generality, always be assumed to be the case), the outcome $a_1$ can be taken as a deterministic function of $x_1$ and $\lambda$, so that $p(a_1 \vert x_1,x_2,\lambda)=p( a_1 \vert x_1,\lambda)$ and $p(a_2 \vert a_1,x_1,x_2,\lambda)=p( a_2 \vert x_1, x_2,\lambda)$.

On the other hand, since $\lambda$ is an ancestor of $x_2$ ($\lambda$ is a parent of $a_1$, who is a parent of $x_2$), one has, in general, that $p(x_2\vert\lambda)\neq p(x_2)$.
However, noting that
\begin{eqnarray}
\label{simplification}
p(x_2 \vert\lambda) & & = \sum_{a_1} p(x_2,a_1\vert\lambda) = \sum_{a_1} p(x_2\vert a_1, \lambda)p( a_1 \vert\lambda) \\ \nonumber
& & = \sum_{a_1} p( x_2 \vert a_1) p(a_1 \vert\lambda),
\end{eqnarray}
we can re-express $p(a_1,a_2 \vert x_1,x_2)$ as
\begin{eqnarray}
\label{wirin2}
& & p(a_1,a_2 \vert x_1,x_2)  = \\ \nonumber
& & \sum_{\lambda,\, a_1^{\prime}}  p(a_1 \vert x_1,\lambda)p( a_2 \vert x_1, x_2,\lambda) p( x_2 \vert a_1^{\prime}) p(a_1^{\prime} \vert \lambda)\frac{1}{p(x_2)} p(\lambda).
\end{eqnarray}
Now, we note that
\begin{align}
\label{eq:use_no_signalling}
\frac{p( x_2 \vert a_1^{\prime})}{p(x_2)}=\frac{p(a_1^{\prime} \vert x_2)}{p(a_1^{\prime})}=1,
\end{align}
where the first equality follows from Bayes' rule and the second one from the nonsignalling constraints (Eq. (1)).
Therefore, Eq. \eqref{wirin2} rewrites
\begin{eqnarray}
& & p(a_1,a_2 \vert x_1,x_2)  = \\ \nonumber
 & & \sum_{\lambda,\, a_1^{\prime}}  p(a_1 \vert x_1,\lambda)p( a_2 \vert x_1, x_2,\lambda) p(a_1^{\prime} \vert \lambda)p(\lambda).
\end{eqnarray}
Finally, using that $\sum_{a_1^{\prime}}p(a_1^{\prime} \vert \lambda)=1$, we arrive at
\begin{eqnarray}
\label{eq:last_proof_lema_1}
& & p(a_1,a_2 \vert x_1,x_2)  = \\ \nonumber
& & \sum_{\lambda}  p(a_1 \vert x_1,\lambda)p( a_2 \vert x_1, x_2,\lambda) p(\lambda).
\end{eqnarray}
This is manifestly the explicit expression of generic correlations produced by Bayesian networks with respect to the IO BDAG $\mathcal{G}_{\rm in-out}$, which finishes the proof.
\end{proof}

The second result we will use states that allowing for a direct causation between any input and $\lambda$ is nonsignaling equivalent to broadcasting the input to the outputs of all $N-1$ other parties [see  Fig. \ref{fig:equiv} b) for the case $N=3$].
\begin{lemma}[Measurement-independence relaxation $\leftrightarrow$ input broadcasting]
\label{lemma:MD}
Let $\mathcal{G}^{(1)}_{\rm mir}$, $\mathcal{G}^{(2)}_{\rm mir}$, and $\mathcal{G}_{\rm ib}$ be any three BDAGs whose only differences are that there exists $1\le i\le N$ such that, for $\mathcal{G}^{(1)}_{\rm mir}$, $\lambda\in{\rm pa}(x_i)$, for $\mathcal{G}^{(2)}_{\rm mir}$, $x_i\in{\rm pa}(\lambda)$, and
 for $\mathcal{G}_{\rm ib}$, $x_i\in{\rm pa}(a_j)$ for all $1\leq j\leq N$. Then, $\mathcal{G}^{(1)}_{\rm mir}$, $\mathcal{G}^{(2)}_{\rm mir}$, and $\mathcal{G}_{\rm ib}$ are nonsignaling equivalent.
\end{lemma}

\begin{proof}
We will prove this lemma explicitly for the particular case of $N=3$ parties. The proof for the general case is totally analogous. The proof strategy consists in showing that, if nonsignalling holds, the most general expression for a correlation compatible with the input-broadcasting  BDAG $\mathcal{G}_{\rm ib}$ in the right-hand panel of Fig. \ref{fig:equiv}
 b) coincides with the most general expressions for correlation compatible with the measurement-dependence BDAGs $\mathcal{G}^{(1)}_{\rm mir}$ and $\mathcal{G}^{(2)}_{\rm mir}$ in the
left-hand and central panels, respectively, of Fig. \ref{fig:equiv}
b). This implies that the set of nonsignalling correlations compatible with $\mathcal{G}_{\rm ib}$ and the set of nonsignalling correlations compatible with $\mathcal{G}^{(1)}_{\rm mir}$ and $\mathcal{G}^{(2)}_{\rm mir}$ are indeed equivalent.

The most general correlation produced by a Bayesian network with respect to $\mathcal{G}_{\rm ib}$, where the setting $x_1$ is a cause of $a_2$, is
\begin{align}
\label{eq:set_broad}
& p(a_1,a_2,a_3 \vert x_1,x_2,x_3) \\ \nonumber  &=  \sum_{\lambda} p(a_1 \vert  x_1,\lambda)\,p(a_2 \vert  x_1,x_2,\lambda)\, p(a_3 \vert  x_1,x_3,\lambda)\,p(\lambda)\\ 
\label{eq:Det_resp_func}
&=\sum_{\lambda_1,\lambda_2,\lambda_3} D^{(1)}_{\lambda_1}( a_1 \vert  x_1)D^{(2)}_{\lambda_2}(a_2\vert x_1, x_2)D^{(3)}_{\lambda_3}(a_3\vert x_1, x_3)\,p(\lambda_1,\lambda_2,\lambda_3).
\end{align}
Eq. \eqref{eq:set_broad} follows from Eq. (2) in the main text with respect to $\mathcal{G}_{\rm ib}$. The right-hand side of  Eq. \eqref{eq:Det_resp_func}, in turn, is simply the expression of the right-hand side of Eq. \eqref{eq:set_broad} in terms of the local deterministic response functions $D^{(i)}_{\lambda_i}$ of each output $a_i$ given its parent inputs with respect to $\mathcal{G}_{\rm ib}$, for $i=1,2,3$, for which we have explicitly decomposed $\lambda$ as the tri-index variable $\lambda=\lambda_1,\lambda_2,\lambda_3$, with $\lambda_i$ labelling the local deterministic strategy of the $i$-th party. More precisely, the local deterministic response functions are defined as
\begin{align}
\label{eq:det_strat_explicit}
& D^{(1)}_{\lambda_1}( a_1 \vert  x_1)\coloneqq\delta_{a_1,f^{(1)}_{\lambda_1}(x_1)}, \\
& D^{(2)}_{\lambda_2}( a_2 \vert  x_1,x_2)\coloneqq\delta_{a_2,g^{(2)}_{\lambda_2}(x_1,x_2)},\\
& D^{(3)}_{\lambda_3}( a_3 \vert  x_1,x_3)\coloneqq\delta_{a_3,g^{(3)}_{\lambda_3}(x_1,x_3)},
\end{align}
with $\delta$ denoting the Kronecker delta, $f^{(1)}_{\lambda_1}$ the $\lambda_1$-th local deterministic assignment of $x_1$ into $a_1$, and $g^{(i)}_{\lambda_i}$ the $\lambda_i$-th local deterministic assignment of $x_1$ and $x_i$ into $a_i$, for $i=2,3$. Note that the deterministic functions $f^{(1)}_{\lambda_1}$ have a single argument, while the deterministic functions $g^{(2)}_{\lambda_2}$ and $g^{(3)}_{\lambda_3}$ have  two arguments.
There are altogether $|\Lambda_1|= |A_1|^{|X_1|}$ different assignments $f^{(1)}_{\lambda_1}$ and $|\Lambda_i|= |A_i|^{|X_1|\times|X_i|}$ different assignments $g^{(i)}_{\lambda_i}$, for $i=2,3$, where $|X_i|$ and $|A_i|$ denote the numbers of inputs and outputs, respectively, of the $i$-th party, for $i=1,2,3$. This gives a total of $|\Lambda|= |\Lambda_1|\times|\Lambda_2|\times|\Lambda_3|$ different global deterministic strategies.

Now, for any fixed $\lambda_i$ and $x_1$, the two-argument assignment $g^{(i)}_{\lambda_i}(x_1,x_i)$, for each $i=1,2$, defines a single-argument deterministic assignment of $x_i$ into $a_i$. Since there are $|\Lambda_i|= |A_i|^{|X_1|\times|X_i|}$ different values of $\lambda_i$ and $|X_1|$ different values of $x_1$, but only $|A_i|^{|X_i|}$ different  deterministic assignments of $x_i$ into $a_i$ possible, $g^{(i)}_{\lambda_i}(x_1,x_i)$ will define the same  deterministic single-argument function of $x_i$ for $\frac{|\Lambda_i|\times|X_1|}{|A_i|^{|X_i|}}=\frac{|A_i|^{|X_1|\times|X_i|}\times|X_1|}{|A_i|^{|X_i|}}=|A_i|^{|X_i|\times(|X_1|-1)}\times|X_1|$
different pairs $(\lambda_i,x_1)$.
Hence, for $i=1,2$, we can introduce $|{\Lambda_i}^{\prime}|\coloneqq|A_i|^{|X_i|}$ different single-argument deterministic functions $f^{(i)}_{{\lambda_i}^{\prime}}$ of $x_i$, explicitly defined by
\begin{align}
\label{eq:def_det_assign}
f^{(i)}_{{\lambda_i}^{\prime}}(x_i)\coloneqq g^{(i)}_{\lambda_i}(x_1,x_i)\text{, with }{\lambda_i}^{\prime}\coloneqq\gamma_i(\lambda_i,x_1),\ \forall\ x_i\in|X_i|,
\end{align}
where we have also introduced $\gamma_i:|\Lambda_i|\times|X_1|\to|{\Lambda_i}^{\prime}|$ as the map that takes all $|\Lambda_i|\times|X_1|$ pairs $(\lambda_i,x_1)$ into the $|{\Lambda_i}^{\prime}|$ different deterministic assignments ${\lambda_i}^{\prime}$ of $x_i$ into $a_i$. Then,  Eqs. \eqref{eq:det_strat_explicit} and \eqref{eq:def_det_assign} imply for the right-hand side of Eq. \eqref{eq:Det_resp_func} that
\begin{widetext}
\begin{align}
\label{eq:set_broad2}
\nonumber
\sum_{\lambda_1,\lambda_2,\lambda_3} D^{(1)}_{\lambda_1}( a_1 \vert  x_1)\,D^{(2)}_{\lambda_2}(a_2\vert x_1, x_2)\,D^{(3)}_{\lambda_3}(a_3\vert x_1, x_3)\,p(\lambda_1,\lambda_2,\lambda_3)
&=\\
\nonumber
\sum_{\substack{\lambda_1,{\lambda_2}^{\prime},{\lambda_3}^{\prime}\\ \lambda_2:\gamma_2(\lambda_2,x_1)={\lambda_2}'\\ \lambda_3:\gamma_3(\lambda_3,x_1)={\lambda_3}'}}
D^{(1)}_{\lambda_1}( a_1 \vert  x_1)\,D^{(2)}_{{\lambda_2}^{\prime}}(a_2\vert x_2)\,D^{(3)}_{{\lambda_3}^{\prime}}(a_3\vert  x_3)\,p(\lambda_1,\lambda_2,\lambda_3)
&=\\
\nonumber
\sum_{\lambda_1,{\lambda_2}^{\prime},{\lambda_3}^{\prime}} D^{(1)}_{\lambda_1}( a_1 \vert  x_1)\,D^{(2)}_{{\lambda_2}^{\prime}}(a_2\vert x_2)\,D^{(3)}_{{\lambda_3}^{\prime}}(a_3\vert  x_3)
\sum_{\substack{\lambda_2:\gamma_2(\lambda_2,x_1)={\lambda_2}'\\ \lambda_3:\gamma_3(\lambda_3,x_1)={\lambda_3}'}}p(\lambda_1,\lambda_2,\lambda_3)
&=\\
\sum_{\lambda_1,{\lambda_2}^{\prime},{\lambda_3}^{\prime}} D^{(1)}_{\lambda_1}( a_1 \vert  x_1)D^{(2)}_{{\lambda_2}^{\prime}}(a_2\vert x_2)D^{(3)}_{{\lambda_3}^{\prime}}(a_3\vert  x_3)\,q(\lambda_1,{\lambda_2}^{\prime},{\lambda_3}^{\prime}\vert x_1),
\end{align}
\end{widetext}
where, in the last equality, we have introduced the normalised conditional probability distribution $q(\lambda_1,{\lambda_2}^{\prime},{\lambda_3}^{\prime}\vert x_1)\coloneqq\sum_{\lambda_2:\gamma_2(\lambda_2,x_1)={\lambda_2}',\lambda_3:\gamma_3(\lambda_3,x_1)={\lambda_3}'}p(\lambda_1,\lambda_2,\lambda_3)$.

\begin{figure}[t!]
\center
\includegraphics[width=.95\columnwidth]{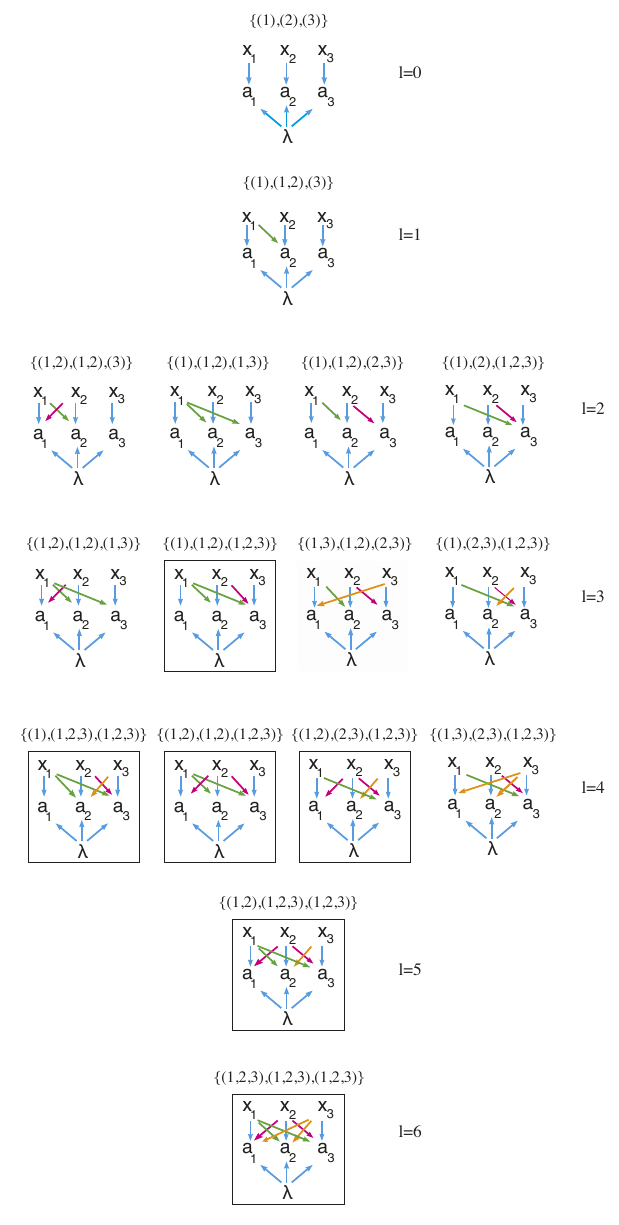}
\caption{Complete causal hierarchy for $N=3$, with 16 causal Bell classes\lo{. Each class is represented by an input-output Bell DAG (IO BDAG), labeled by a set $\left\{\boldsymbol{{\rm in}}_1,\boldsymbol{{\rm in}}_2,\boldsymbol{{\rm in}}_3 \right\}$ of vectors $\boldsymbol{{\rm in}}_i$ composed of the parent inputs of $a_i$. 
Each level of the hierarchy is defined by the total number $l$ of input-to-output locality relaxations. For example,
$\left\{(1),(2),(1,2,3)\right\}$, in level $l=2$, represents the star structure, where one output is causally influenced by all inputs.} The IO BDAGs inside the black boxes represent nonsignaling boring classes, i.e., that they can reproduce all nonsignalling correlations \lo{(see App.} \ref{sec:NS_interesting}\lo{)}.}
\label{fig:belleq}
\end{figure}

The right-hand side of the last line of Eq.~\eqref{eq:set_broad2} is readily identified as the decomposition into deterministic strategies of the most general correlation produced by a Bayesian network with respect to $\mathcal{G}^{(\lo{2})}_{\rm mir}$. This, in turn, is actually also equivalent to the most general correlation produced by a Bayesian network with respect to $\mathcal{G}^{(\lo{1})}_{\rm mir}$. This can be immediately seen by noting that, instead of $q(\lambda_1,{\lambda_2}^{\prime},{\lambda_3}^{\prime}\vert x_1)$ in Eq.~\eqref{eq:set_broad2}, $\mathcal{G}^{(\lo{1})}_{\rm mir}$ would naturally give $\frac{p(x_1\vert \lambda_1,{\lambda_2}^{\prime},{\lambda_3}^{\prime})\, p(\lambda_1,{\lambda_2}^{\prime},{\lambda_3}^{\prime})}{p(x_1)}$ . However, the two expressions are trivially equal due to Bayes' theorem.
\end{proof}

\lo{L}emmas \lo{\ref{lemma:input-to-output} and \ref{lemma:MD}} imply that every causal relaxation on a LHV model is, as for what nonsignaling correlations concerns, accounted for (in the inclusion sense) by an input-to-output locality relaxation. We refer to any BDAG whose only causal relaxations consist of input-to-output locality relaxations as an \emph{input-output (IO) BDAG}.
Every IO BDAG can be defined by the subsets of inputs that are parents of each ouput (see Fig. \ref{fig:hierarchy} for more details). We emphasise that, as a consequence of the lemmas, the total number of BDAGs to scrutinise is hugely reduced. Namely, there are 15 different ways of drawing directed edges from one party to another [all the particular instances of the general locality relaxation of Fig. \ref{fig:equiv} a)]. All corresponding 15 BDAGS are grouped together with a single IO BDAG due to lemma \ref{lemma:input-to-output}. Furthermore, each BDAG with directed edges from $\lambda$ to any of the inputs is grouped together with an IO BDAG due to lemma \ref{lemma:MD}.
Hence, IO BDAGs make generic representatives of all possible causal relaxations in the nonsignaling framework.
This leads us to the following natural classification.
\begin{definition}[Causal classes of Bell correlations]
\label{def:Bell_classes}
Each IO BDAG $\mathcal{G}_{\rm in-out}$ defines a \emph{causal class of Bell correlations}, or, for short, a \emph{causal Bell class}, as the convex hull of nonsignaling correlations produced by Bayesian networks with respect to $\mathcal{G}_{\rm in-out}$ or any of its party-permutation equivalents.
In addition, we call a causal Bell class \emph{nonsignaling interesting} if there exist nonsignaling correlations incompatible with it; otherwise we call it \emph{nonsignaling boring} \footnote{The terminology nonsignalling interesting/boring is inspired by the works \cite{HLP14,Pienaar16}. There, the terms interesting/boring are used to denote DAGs that allow/do not allow one to distinguish -- though in a slightly different sense as here -- between classical, quantum, and post-quantum models}.
Finally, each nonsignaling interesting causal Bell class defines a \emph{class of multipartite Bell nonlocality}, as the set of all nonsignaling correlations outside the causal Bell class.
\end{definition}

\begin{figure}[t!]
\vspace{0.6cm}
\center
\includegraphics[width=.9\columnwidth]{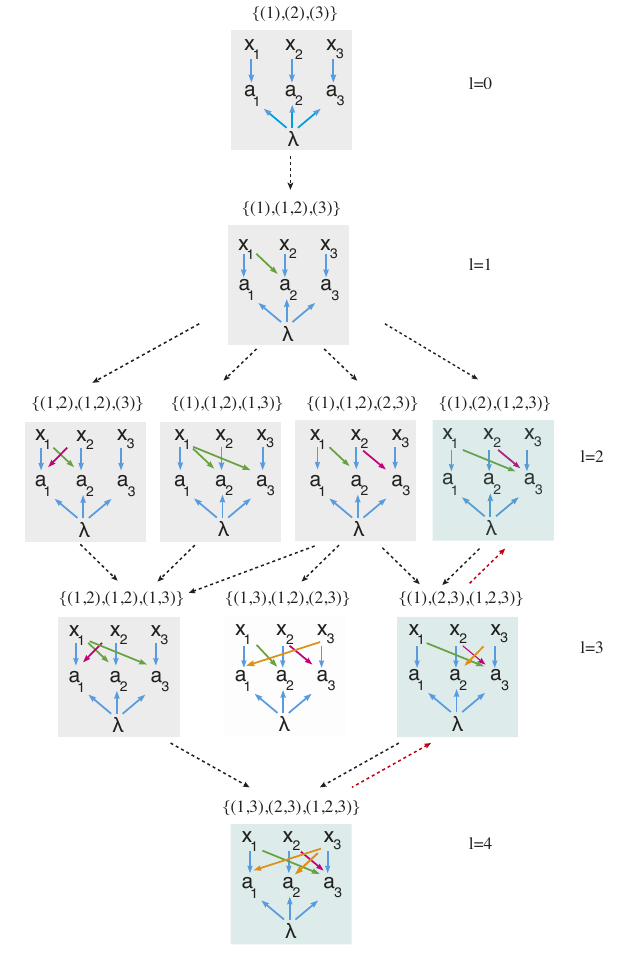}
\caption{ Hierarchy of nonsignaling interesting causal classes of Bell correlations for $N=3$ \lo{(corresponding to the IO BDAGs not inside black boxes in Fig. \ref{fig:belleq}).}
Every dashed arrow from one DAG \lo{(in a given level)} to another \lo{DAG (in a different level)} indicates that
the latter nonsignaling implies the former.
Black dashed arrows represent the implications that the hierarchy automatically imposes, whereas red ones implications that we prove by other means (see text and appendix).
The 6 light-grey shaded classes were known not to reproduce all quantum correlations \cite{Jones2005}.
From the remaining four classes, the three light-green shaded ones collapse to the star class \lo{$\left\{(1),(2),(1,2,3)\right\}$} (see red arrows). We find quantum correlations beyond this class.}
\label{fig:hierarchy}
\end{figure}

This characterisation offers, as mentioned, a physically meaningful tree-like hierarchy, where classes in a given level are nonsignalling implied by classes in the level below. We refer to it as the \emph{causal hierarchy},  represented in \figref{fig:hierarchy} for the $N=3$. We note that only the nonsignalling interesting part of the hierarchy is plotted in the figure. The complete tripartite causal hierarchy \lo{is graphically represented in Fig. \ref{fig:belleq}. It} contains a total of 16 causal Bell classes, but 6 of them are nonsignalling boring (see also the appendix for a list of all the classes in the complete hierarchy for $N=4$). Furthermore, apart from the above-mentioned implications that follow automatically from the hierarchy, other nonsignalling implications can take place. For $N=3$, for instance, 3 of the 10 nonsignalling interesting classes turn out to be equivalent, thus
collapsing to a single class.
All this is formalised by the following theorem, proven in the appendix.
\begin{theorem}[Classes of tripartite nonlocality]
\label{theo:tripartite_classes}
For $N=3$, there are 10 nonsignalling interesting causal Bell classes, represented by the IO BDAGs of \figref{fig:hierarchy}. 8 out of these 10 are inequivalent, each one yielding a different class of Bell nonlocality. Furthermore, at least 7 different nonlocality classes contain quantum correlations.
\end{theorem}

\section{Stronger forms of quantum nonlocality}
 The 6 IO BDAGs shaded in light grey in Fig. \ref{fig:hierarchy} belong to a class of models that were shown \cite{Jones2005} to satisfy Svetlichny's inequality, which can be violated by quantum correlations \cite{Svetlichny1987}. The remaining 4 IO BDAGs define causal Bell classes for which no quantum violation was known so far. The 3 of them shaded in light green in Fig. \ref{fig:hierarchy} are nonsignalling equivalent, as mentioned, so that the fourth level, and part of the third one, of the hierarchy collapse to the second level \footnote{In fact, this is a generic property that holds for arbitrary $N$: all causal Bell classes for which all $N$ inputs go to the output of one party while any other locality relaxation involves the latter party are nonsignalling equivalent (see appendix)}.
 We refer to the resulting class as the \emph{star}, because one party (the centre of the star) receives the inputs of all others parties (the rays). Remarkably, we find quantum correlations outside of it.

In the appendix, we show that the star class is nonsignalling boring for the specific scenario of 3 parties with 2 inputs and 2 outputs each.
However, it satisfies a broad family of  Bell inequalities that are non-trivial for higher output alphabets. Consider output alphabets where each output can be factorised into two integer variables. Then, the inequalities in question can be expressed in a unified fashion as
\begin{equation}
\label{sum_generic}
I_{3}\coloneqq I_{2}(A,B)+I_{2}(A^{\prime},C)+I_{2}(B^{\prime},C^{\prime}) \leq \beta_{\rm L}+2\,\beta_{\rm NS},
\end{equation}
where $I_{2}$ stands for an arbitrary bipartite Bell expression with LHV and nonsignalling bounds $\beta_{\rm L}$ and  $\beta_{\rm NS}$, respectively. $A$ and $A^{\prime}$ are the two variables that encode the output of the first party, $B$ and $B^{\prime}$ the output of the second one, and $C$ and $C^{\prime}$ that of the third one. For example, for 2 inputs and 4 outputs per part, $I_2$ can be taken as the Clauser-Shimony-Holt inequality \cite{Clauser1969}, with $\beta_{\rm L}=2$ and $\beta_{\rm NS}=4$ (\lo{note that} this inequality has also been discussed, in a different context, in \lo{Ref.} \cite{Vertesi2012}). Then, $A$ and $A^{\prime}$ are bits, generated, for instance, by making the same measurement on two independent subsystems, and equivalently for $B$, $B^{\prime}$, $C$ and $C^{\prime}$. In that case, the resulting tripartite inequality $I_3$
can be maximally violated by three bipartite boxes independently distributed among the parties. More precisely, we refer to the well-known PR boxes \cite{Popescu1994}, which are post-quantum but nonsignalling.

Surprisingly, for higher dimensions, there exist choices of $I_2$ for which $I_3$ is maximally violated in quantum mechanics. Consider for instance $I_2$ as the all-versus-nothing Bell inequality studied in Refs. \lo{\cite{GMS07,Yang05,Aolita11}} (containing 3 inputs and 4 outputs per party). This  bipartite inequality \lo{is tight \cite{GMS07} and} can be violated up to its algebraic \lo{maximum} by correlations obtained from a pair of singlets \lo{ \cite{GMS07,Yang05,Aolita11}}. Hence, $I_3$ can be maximally violated with 6 singlets (2 singlets per pair of parties, same measurement on both pairs of qubits of each party). \lo{Since here each party holds two subsystems (\eg $A$ and $A'$), in our case each party will have 16 outcomes.} Another alternative is to choose $I_2$ as the chained inequality, which requires a single singlet per pair but a very large number of inputs per party \cite{BKP06} (see the appendix). Importantly, for both choices, experiments with violations of $I_2$ high enough to violate $I_3$ have been demonstrated  \cite{Yang05,Aolita11,Stuart12,Christensen15}.

\section{A non-trivial Bell-type inequality without quantum violation}
The last nonsignalling interesting class to analyse is $\{(1,3),(1,2),(2,3)\}$ in the third level, where each party communicates his/her input to his/her nearest neighbour in a circle-like configuration. We refer to this class as the \emph{circle}. Interestingly, correlations in this class attain a unit success probability of winning the GYNI nonlocal game of Ref. \cite{Almeida2010}.
We have not found quantum violations of this class; but we found that it satisfies the binary-input-binary-output inequality:
\begin{align}
\label{new_GYNI}
\nonumber
\mean{A_0\,B_0\,C_0}+\mean{A_0\,B_0\,C_1}+\mean{A_0\,B_1\,C_0}+\mean{A_0\,B_1C_1} &+\\
\mean{A_1\,B_0\,C_0}+\mean{A_1\,B_0\,C_1}+\mean{A_1\,B_1\,C_0}-\mean{A_1\,B_1\,C_1} &\leq 6.
\end{align}
This is maximally violated by the nonsignaling extremal correlations $p(a_1,a_2,a_3 \vert x_1,x_2,x_3)= \delta_{a_1\oplus a_2 \oplus a_3,x_1\times x_2\times x_3}/4$, where $\delta$ denotes the Kronecker delta and $ \oplus$ addition modulo 2. Using the techniques of Refs. \cite{NPA1,NPA2}, we see that Eq. \eqref{new_GYNI} is not violated by any quantum correlations. Thus, Eq. \eqref{new_GYNI} constitutes a non-trivial Bell-type inequality with no quantum violation. Note, nevertheless, that this does not imply that the circle class contains all quantum correlations, as there might be other inequalities (involving not only full-correlators) that admit quantum violations. See the appendix for details.

\section{Discussion}
We proposed a hierarchical classification for all the relaxations of locality and measurement independence in Bell's theorem in terms of the nonsignalling correlations to which they lead. The nonsignalling correlations compatible with an arbitrary causal structure are always captured by a (typically much simpler) causal network involving only locality relaxations where the input of one party causally influences the outputs of others. The framework facilitates the study of unexplored forms of multipartite Bell nonlocality. For instance, we identified new tripartite causal structures that cannot reproduce all quantum correlations. This demonstrates the strongest form of quantum multipartite nonlocality known, closing a long-standing open question \cite{Jones2005}. Furthermore, as another application, we derived a previously unknown non-trivial Bell-type inequality without a quantum violation.

Our work offers a number of exciting questions for future research. In particular, the discovery of  new and stronger forms of quantum correlations offers a vast, unexplored territory.
In addition, from a fundamental perspective, the fact that our framework naturally leads to a non-trivial Bell-type inequality with no quantum violation is appealing \cite{Winter2010}. From an applied perspective, our results may have implications in communication complexity problems \cite{Buhrman2010} or in the emerging field of quantum causal networks \cite{Leifer2013,Henson2014,Fritz2014,Chaves2015a,Piennar2015,Ried15,Costa2016}.
In conclusion, we believe our findings can open a new chapter in the understanding of multipartite Bell nonlocality.

\begin{acknowledgments}
We would like to specially thank Fernando de Melo for the hospitality at Rio de Janeiro's CBPF, where the first ideas that led to this paper were conceived, as well as S. Pironio and C. Branciard for useful comments. RC acknowledges financial support from the Brazilian ministries MEC and MCTIC, the FQXi Fund, the Excellence Initiative of the German Federal and State Governments (Grants ZUK 43 \& 81), the US Army Research Office under contracts W911NF-14-1-0098 and W911NF-14-1-0133 (Quantum Characterization, Verification, and Validation), the DFG (GRO 4334 \& SPP 1798). DC acknowledges the Ram\'on y Cajal fellowship, the Spanish MINECO (Severo Ochoa grant SEV-2015-0522), and the AXA Chair in Quantum Information Science. LA acknowledges financial support from the Brazilian agencies CNPq, FAPERJ, CAPES, and INCT-IQ.
\end{acknowledgments}

%

\begin{widetext}

\appendix


\section{Characterization of the causal polytope of each IO BDAG}
\label{sec:caus_poly}
We denote the number of inputs and outputs of the $i$-th party by $|X_i|$ and $|A_i|$, respectively, for $i=1, \hdots N$. Then, every correlations $\boldsymbol{p}$ can be thought of as a vector, with components $\boldsymbol{p}_{a_1,\dots,a_N, x_1,\dots,x_N}\coloneqq p(a_1,\dots,a_N \vert x_1,\dots,x_N)$, in the real vector space of dimension $d_{\boldsymbol{A}\vert\boldsymbol{X}}=|A_1|\times \hdots |A_N|\times |X_1|\times \hdots |X_N|$, i.e. $\RR^{d_{\boldsymbol{A}\vert\boldsymbol{X}}}$. If $\boldsymbol{p}$ is compatible with an arbitrary IO BDAG $\{\boldsymbol{{\rm in}}_1, \hdots , \boldsymbol{{\rm in}}_N\}$, where $|\boldsymbol{{\rm in}}_i|$ denotes the number of parent inputs of the $i$-th output, it can, without loss of generality, be decomposed as
\begin{equation}
\label{p_gen}
p(a_1,\dots,a_N \vert x_1,\dots,x_N)= \sum_{\lambda} D^{(1)}_{\lambda_1}(a_1 \vert  \boldsymbol{{\rm in}}_1)\hdots D^{(N)}_{\lambda_N}(a_N\vert\boldsymbol{{\rm in}}_N)\,q(\lambda),
\end{equation}
with $q(\lambda)\geq0$ and $\sum_{\lambda}q(\lambda)=1$. Here,
in a similar fashion to in the previous section, the multi-variable decomposition $\lambda=\lambda_1,\hdots,\lambda_N$ is used and $D^{(i)}_{\lambda_i}$ stands for the local deterministic response function of the $i$-th output $a_i$ given the vector $\boldsymbol{{\rm in}}_i$ of its parent inputs for the $\lambda_i$-th local deterministic strategy,  for $i=1, \hdots,N$. Each $D^{(i)}_{\lambda_i}$ is explicitly given by
\begin{equation}
\label{eq:local_det_st}
D^{(i)}_{\lambda_i}(a_i \vert  \boldsymbol{{\rm in}}_i)\coloneqq \delta_{a_i,f^{(i)}_{\lambda_i}(\boldsymbol{{\rm in}}_i)},
\end{equation}
with $\delta$ denoting the Kronecker delta and $f^{(i)}_{\lambda_i}$ being the $\lambda_i$-th local deterministic assignment of $\boldsymbol{{\rm in}}_i$ into $a_i$. In addition, the $\boldsymbol{\lambda}$-th global deterministic response function is given by the product
\begin{equation}
\label{eq:global_det_st}
\boldsymbol{D}_{\lambda}\coloneqq D^{(1)}_{\lambda_1}\times \hdots D^{(N)}_{\lambda_N}.
\end{equation}
Each $\boldsymbol{D}_{\lambda}$ is clearly also a vector in $\RR^{d_{\boldsymbol{A}\vert\boldsymbol{X}}}$, with components ${\boldsymbol{D}_{\lambda}}_{a_1,\dots,a_N, x_1,\dots,x_N}\coloneqq\boldsymbol{D}_{\lambda}(a_1,\dots,a_N \vert x_1,\dots,x_N)= D^{(1)}_{\lambda_1}(a_1 \vert  \boldsymbol{{\rm in}}_1)\times \hdots D^{(N)}_{\lambda_N}(a_N \vert  \boldsymbol{{\rm in}}_N)$.
Thus, $\boldsymbol{p}$ actually lives in a polytope in $\RR^{d_{\boldsymbol{A}\vert\boldsymbol{X}}}$ defined by the convex hull of a finite number $|\Lambda|$ of extremal points, each extremal point given by the vector $\boldsymbol{D}_{\lambda}$ for a different $\lambda$.
There are $|\Lambda_i|\coloneqq |A_i|^{|X_i|^{|\boldsymbol{{\rm in}}_N|}}$ different local deterministic strategies for each $i=1, \hdots,N$. So the total number of different global deterministic strategies is $|\Lambda|\coloneqq\prod_{i=1}^N |\Lambda_i|=\prod_{i=1}^N |A_i|^{|X_i|^{|\boldsymbol{{\rm in}}_i|}}$.

It is convenient to identify each extremal vector $\boldsymbol{D}_{\lambda}$ with the $\lambda$-th column of a $d_{\boldsymbol{A}\vert\boldsymbol{X}}\times|\Lambda|$ real matrix $\boldsymbol{D}$, with components $\boldsymbol{D}_{a_1,\dots,a_N, x_1,\dots,x_N,\lambda}\coloneqq\boldsymbol{D}_{\lambda}(a_1,\dots,a_N \vert x_1,\dots,x_N)$, and each $q(\lambda)$ as the $\lambda$-th component of a $|\Lambda|$-dimensional real vector $\boldsymbol{q}$. With this, Eq. \eqref{p_gen} can be rewritten concisely as
\begin{equation}
\label{p_concise}
\boldsymbol{p}=\boldsymbol{D}\cdot\boldsymbol{q},
\end{equation}
where  the symbol $\cdot$  stands for  contraction over the index $\lambda$. The tensor $\boldsymbol{D}$ of deterministic strategies depends exclusively  on the number of inputs and outputs per party, as well as on the the causal structure in question. It characterises completely the polytope of all correlations (both signalling and nonsignalling) compatible with the IO BDAG $\{\boldsymbol{{\rm in}}_1, \hdots , \boldsymbol{{\rm in}}_N\}$. We call such polytope the \emph{causal polytope} of $\{\boldsymbol{{\rm in}}_1, \hdots , \boldsymbol{{\rm in}}_N\}$. In contrast, the vector $\boldsymbol{q}$ is in one to one correspondence with the particular $\boldsymbol{p}$.

Hence, the problem of determining whether a given $\boldsymbol{p}$ is compatible with a Bayesian network with respect to $\{\boldsymbol{{\rm in}}_1, \hdots , \boldsymbol{{\rm in}}_N\}$ is equivalent to determining whether there exists $\boldsymbol{q}$ such that Eq. \eqref{p_concise} holds. This, since Eq. \eqref{p_concise} defines a system of linear equations, can always be solved efficiently in the length $|\Lambda|$ of the vector $\boldsymbol{q}$. A practical tool to do this is linear programming. More precisely, solving  the linear system given by Eq. \eqref{p_concise} is equivalent to solving the linear programme
\begin{eqnarray}
\label{LP}
\nonumber
&& \ \ \ \ \  \text{Given } \boldsymbol{D} \text{ and } \boldsymbol{p},  \\
\nonumber
&& \underset{\boldsymbol{q} \in \RR^{|\Lambda|},\, \boldsymbol{q}\geq0,\, \|\boldsymbol{q}\|=1}{\minimize} \, \boldsymbol{I}\cdot \boldsymbol{q},\\
&&\ \ \ \ \  \st  \boldsymbol{D}\cdot\boldsymbol{q}=\boldsymbol{p},
\end{eqnarray}
with $\boldsymbol{q}\geq0$ and $\|\boldsymbol{q}\|=1$ short-hand notations for $q(\lambda)\geq0$, for all $\lambda=1, \hdots |\Lambda|$, and $\sum_{\lambda}q(\lambda)=1$, respectively, and where  $\boldsymbol{I}$ is any vector in $\RR^{|\Lambda|}$ (that encodes the so-called objective function). If the linear programme \eqref{LP} is feasible, $\boldsymbol{p}$ is compatible with  $\{\boldsymbol{{\rm in}}_1, \hdots , \boldsymbol{{\rm in}}_N\}$. Otherwise $\boldsymbol{p}$ is not inside the causal polytope of $\{\boldsymbol{{\rm in}}_1, \hdots , \boldsymbol{{\rm in}}_N\}$. In turn, using standard convex-optimization tools, such as for instance the software PORTA \cite{porta}, one can also find the dual description of the polytope in terms not of its extremal points but of its facets, i.e., its Bell inequalities.

Finally, we recall that the causal Bell class associated to $\{\boldsymbol{{\rm in}}_1, \hdots , \boldsymbol{{\rm in}}_N\}$ is actually defined (see Def. 3 in the main text) not by all the correlations compatible with it but by the convex hull of all nonsignalling correlations compatible with any IO BDAG obtained via party exchanges from it. In that case, one proceeds in a similar way but taking into account all the different global deterministic-strategy tensors arising from party-permutations of $\{\boldsymbol{{\rm in}}_1, \hdots , \boldsymbol{{\rm in}}_N\}$ and adding to the constraints of the linear programme the nonsignaling constraints on $\boldsymbol{p}$, given by Eq. (1) in the main text.

\section{Proof of Theorem 4}
\label{Proof_tripartite}

Our proof of Theorem 4 in the main text consists of 4 steps:

\smallskip
\textbf{1.} Listing all the different tripartite IO BDAGs (without considering party permutations of each of them). For $N=3$ there are 16 such.

\smallskip
\textbf{2.} Proving that 6 of them define nonsignalling boring causal Bell classes.

\smallskip
\textbf{3.} Proving that the 10 remaining classes are all nonsignalling interesting and that at least 7 of them cannot reproduce all quantum correlations.

\smallskip
\textbf{4.} Proving that, out of the 10 nonsignalling interesting classes, the number of inequivalent ones is 8.

\smallskip
In the following 4 subsections we describe each step in detail.

\subsection{The complete causal hierarchy (i.e., including nonsignaling boring classes)}
For arbitrary $N$, the complete causal hierarchy contains a total of $N\,(N-1)$ levels, including the LHV zeroth level.
The $[N\,(N-1)]$-th level is occupied by the IO BDAG $\{(1, \hdots, N), \hdots, (1, \hdots, N)\}$, for which all $N$ inputs are parents of all $N$ outputs. The latter causal structure can reproduce all correlations (including the signalling ones), so that the corresponding causal Bell class is trivially nonsignalling boring. For each fixed number $l$ of input-output locality relaxations, there are a total of $\binom{N\,(N-1)}{l}$ different IO BADGs, but many of them are redundant, as they are equivalent up to party exchanges. Eliminating all the party-exchange redundancies leaves us with the IO BDAGs that define the complete causal hierarchy.
Following the procedure above for $N=3$ yields 16 different non-redundant IO BDAGs, graphically represented in \figref{fig:belleq}. In addition, in table \ref{table_3}, we summarise the main properties of each of these IO BDAGs. (See also Sec. \ref{sec:Four-partite} for a brief description of the complete causal hierarchy for the four-partite case.)
\begin{table}[t]
  \begin{tabular}{|c|c|c|c|c|}
    \hline
Level & Representative IO BDAG & Party-exchange symmetries & Quantum violation & Nonsignaling violation \\ \hline
0	 & 	\{(1),(2),(3)\}	 & 	1	 & 	Yes	& 	Yes		 	 	 \\ \hline
1	 & 	\{(1),(1,2),(3)\}	 & 	6	 & 	Yes	& 	Yes		  	 	 \\ \hline

2	 & 	\{(1,2),(1,2),(3)\}	 & 	3	 & 	Yes	& 	Yes		 	 	 \\
	 & 	\{(1),(1,2),(1,3)\}	 & 	6	 & 	Yes	& 	Yes		 	 	 \\
	 & 	\{(1),(1,2),(2,3)\}	 & 	3	 & 	Yes	& 	Yes		 	 	 \\
	 & 	\{(1),(2),(1,2,3)\}	 & 	3	 & 	Yes & 	Yes			 	 	 \\ \hline

3	 & 	\{(1,2),(1,2),(1,3)\}	 & 	6	 & 	Yes & 	Yes			 	 	 \\
	 & 	\{(1),(1,2),(1,2,3)\}	 & 	6	 &  No &  No			 	 	 \\
	 & 	\{(1),(2,3),(1,2,3)\}	 & 	6	 & 	Yes & 	Yes		 	 	 \\
	 & 	\{(1,3),(1,2),(2,3)\}	 & 	2	 & 	Unknown & 	Yes			 	 	 \\ \hline

4	 & 	\{(1),(1,2,3),(1,2,3)\}	 & 	3	 &  No &  No			 	 	 \\
	 & 	\{(1,2),(1,2),(1,2,3)\}	 & 	3	 &  No &  No			 	 	 \\
	 & 	\{(1,2),(2,3),(1,2,3)\}	 & 	6	 &  No &  No			 	 	 \\
	 & 	\{(1,3),(2,3),(1,2,3)\} & 	3	 & 	Yes & 	Yes			 	 	 \\ \hline

5	 & 	\{(1,2),(1,2,3),(1,2,3)\}	 & 	6	 & No	&  No			 	 	 \\ \hline
6	 & 	\{(1,2,3),(1,2,3),(1,2,3)\}	 & 	1	 & 	No &  No			 	 	 \\ \hline
      \end{tabular}
\caption{Main properties of the 16 causal Bell classes that compose the complete  causal hierarchy for $N=3$. The first column indicates the level of the hierarchy. In the second column, the representative IO BDAGs of each class (those protted in Fig. \ref{fig:belleq}) are shown. In the third column, the total numbers of party-permutation symmetries of each IO BDAG are shown. The fourth and fifth column indicate which causal Bell classes (which take into account all the corresponding party-exchange symmetries, according to Def. 3 in the main text) are violated by quantum and nonsignalling correlations, as shown in Secs. \ref{sec:7_quantum}.
}
\label{table_3}
\end{table}

\subsection{Five nonsignalling boring classes in the tripartite scenario}
\label{sec:NS_interesting}
In this subsection, we prove that the 5 causal Bell classes represented by the IO BDAGs in black boxes in Fig. \ref{fig:belleq} are nonsignalling boring, i.e., they can reproduce all nonsignalling correlations. We do this by explicitly proving that the class $\{(1),(1,2),(1,2,3)\}$ in the third level is nonsignalling boring. This automatically implies that the other 4 classes (3 in the fourth level and the one in the fifth level) are nonsignalling boring too, as they can all be obtained from $\{(1),(1,2),(1,2,3)\}$ by causal relaxations.

\begin{proof}[Proof that $\{(1),(1,2),(1,2,3)\}$ is nonsignalling boring]

Consider arbitrary tripartite correlations $\boldsymbol{p}$ with elements $p(a_1,a_2,a_3 \vert x_1,x_2,x_3)$. Without loss of generality, these can be decomposed as
\begin{eqnarray}
\label{eq:subsec2_1}
p(a_1,a_2,a_3 \vert x_1,x_2,x_3)&=&p(a_3 \vert x_1,x_2,x_3,a_1,a_2)\,p(a_2\vert x_1,x_2,x_3,a_1)\,p(a_1\vert x_1,x_2,x_3)\\
\label{eq:subsec2_2}
&=&p(a_3 \vert x_1,x_2,x_3,a_1,a_2)\,p(a_2\vert x_1,x_2,a_1)\,p(a_1\vert x_1),
\end{eqnarray}
where Eq. \eqref{eq:subsec2_1} follows from Bayes' rule and Eq. \eqref{eq:subsec2_2} from the nonsignalling constraints, given by Eq. (1) in the main text. Now, from Eq. (2) in the main text, we see that $p(a_3 \vert x_1,x_2,x_3,a_1,a_2)\,p(a_2\vert x_1,x_2,a_1)\,p(a_1\vert x_1)$ is the explicit expression of  correlations produced by a Bayesian network with respect to a BDAG with locality relaxations from $A$ to $B$, $A$ to $C$, and $B$ to $C$. By virtue of Lema 1 in the main text, such correlations are always within the causal Bell class $\{(1),(1,2),(1,2,3)\}$.
\end{proof}

\subsection{10 nonsignalling interesting tripartite classes, at least 7 of which with quantum violations}
\label{sec:7_quantum}

In this subsection, we prove that the remaining 10 classes are nonsignalling interesting. We do that by deriving Bell inequalities for each causal Bell class that are violated by nonsignalling correlations. Furthermore, for 7 of the classes, the violations that we find are not only nonsignalling but actually quantum.

The classes represented by $\{(1),(2),(3)\}$, $\{(1),(1,2),(3)\}$, $\{(1,2),(1,2),(3)\}$, $\{(1),(1,2),(1,3)\}$, $\{(1),(1,2),(2,3)\}$ and $\{(1,2),(1,2),(1,3)\}$ define, in the terminology of Ref. \cite{Jones2005}, \emph{partially paired} correlations. In Ref. \cite{Jones2005}, it was shown that all partially paired correlation, respect the Svetlichny inequality \cite{Svetlichny1987}
\begin{eqnarray}\label{eq: svet ineq_app}
-\mean{A_0B_0C_0}+\mean{A_0B_0C_1}+\mean{A_0B_1C_0}+\mean{A_0B_1C_1} +\mean{A_1B_0C_0}+\mean{A_1B_0C_1}+\mean{A_1B_1C_0}-\mean{A_1B_1C_1} \leq 4.
\end{eqnarray}
This inequality is
violated by quantum correlations obtained from local measurements on entangled quantum states, the maximum quantum violation being $4\sqrt{2}$, with Greenberger-Horne-Zeilinger states \cite{Svetlichny1987}. Thus, the six classes are not only nonsignalling interesting but they also admit quantum violations.

For the circle class, represented by $\{(1,3),(1,2),(2,3)\}$, we derive a previously unknown non-trivial tight Bell inequality for full correlators. The set of full correlators compatible with a given class defines also a polytope. Thus, as discussed in Sec. \eqref{sec:caus_poly}, standard convex optimization tools \cite{porta} can be used to obtain the Bell inequalities for full correlators.
Notice also that, given a Bell inequality, a simple way to see if it is satisfied by a causal Bell class is to check that all global deterministic strategies of the class (see Sec. \ref{sec:caus_poly}) respect the inequality.
We find that all full correlators compatible with the circle class satisfy  the inequality
\begin{eqnarray}
\label{new_GYNI_app}
\mean{A_0B_0C_0}+\mean{A_0B_0C_1}+\mean{A_0B_1C_0}+\mean{A_0B_1C_1} +\mean{A_1B_0C_0}+\mean{A_1B_0C_1}+\mean{A_1B_1C_0}-\mean{A_1B_1C_1} \leq 6
\end{eqnarray}
This inequality is violated up to the algebraic maximal value 8 by the nonsignaling extremal correlations
\begin{equation}
p(a_1,a_2,a_3 \vert x_1,x_2,x_3)= \frac{1}{4}\,\delta_{a_1\oplus a_2 \oplus a_3,\,x_1\times x_2\times x_3},
\end{equation}
originally identified in Ref. \cite{Barrett2005}. This proves that $\{(1,3),(1,2),(2,3)\}$ is nonsignalling interesting.
Using the techniques of Refs. \cite{NPA1,NPA2}, we see that Eq. \eqref{new_GYNI} is not violated by any quantum correlations. Thus, Eq. \eqref{new_GYNI} constitutes a novel non-trivial Bell inequality with no quantum violation. Note, nevertheless, that this does not imply that the causal Bell class $\{(1,3),(1,2),(2,3)\}$ can reproduce all quantum correlations, as there might be other inequalities (involving not only full-correlators) that admit quantum violations. \dani{Finally, we were unable to determine if this inequality is tight in the space of probability distributions defined by the class due to the complexity of determining all the extremal points of this polytope.}

The three remaining  classes, represented by $\{(1),(2),(1,2,3)\}$, $\{(1),(2,3),(1,2,3)\}$ and $\{(1,3),(2,3),(1,2,3)\}$, are
 found to be nonsignalling boring for the restricted case of binary inputs and outputs. To see this, we solved, once again with standard convex-optimization tools \cite{porta}, the feasibility problem of Eq. \eqref{p_concise} for all the 46  extremal nonsignalling correlations \cite{Pironio2011} for the binary-input binary-output case, the same ones used for table \ref{table_noise}. All 46 extremal points are found to belong to the causality polytopes of $\{(1),(2),(1,2,3)\}$, $\{(1),(2,3),(1,2,3)\}$ and $\{(1,3),(2,3),(1,2,3)\}$. This is however not true for higher numbers of inputs and outputs, as we next show analytically.

The crucial observation is to note that any (tripartite) correlations $\boldsymbol{p}$ produced by a Bayesian network with respect to any of the three IO BDAGs, $\{(1),(2),(1,2,3)\}$, $\{(1),(2,3),(1,2,3)\}$ or $\{(1,3),(2,3),(1,2,3)\}$, has marginal (bipartite) correlations over Alice and Bob, the first and second parties, respectively, described by a bipartite LHV model. The intuitive explanation for this is that in none of the three BDAGs there are arrows going from Alice to Bob or from Bob to Alice. In the end of this subsection, we prove this fact formally. This fact implies that the three corresponding causal Bell classes consist exclusively of convex combinations of (tripartite) correlations each of which has a LHV bipartite marginal over some pair out of the three parties.
From this, in turn, it follows that the three causal Bell classes satisfy a broad family of non-trivial Bell inequalities that can all be described in a unified way by the generic expression
\begin{equation}
\label{sum_generic_app}
I_{3}\coloneqq I_{2}(A,B)+I_{2}(A^{\prime},C)+I_{2}(B^{\prime},C^{\prime}) \leq \beta_{\rm L}+\beta_{\rm NS}+\beta_{\rm NS},
\end{equation}
where $I_{2}$ stands for any arbitrary bipartite linear Bell expression with local bound $\beta_{\rm L}$ and nonsignalling bound $\beta_{\rm NS}$. $A$ and $A^{\prime}$ are random variables associated to the outputs of Alice, $B$ and $B^{\prime}$ to the outputs of Bob, and $C$ and $C^{\prime}$ to those of Charlie, the third party.

For instance, the simplest non-trivial example we find is in the scenario of 2 inputs and 4 outputs per party. There, each output can, without loss of generality, be represented by two bits: $A$ and $A^{\prime}$ for Alice, $B$ and $B^{\prime}$ for Bob, and $C$ and $C^{\prime}$ for Charlie. Then, $I_{2}$ can  be chosen as the Clauser-Shimony-Holt inequality \cite{Clauser1969}: $I_{2}(A,B)=CHSH(A,B)\coloneqq\mean{A_0\,B_0}+\mean{A_0\,B_1}+\mean{A_1\,B_0}-\mean{A_1\,B_1}$, $I_{2}(A^{\prime},C)=CHSH(A^{\prime},C)\coloneqq\mean{A^{\prime}_0\,C_0}+\mean{A^{\prime}_0\,C_1}+\mean{A^{\prime}_1\,C_0}-\mean{A^{\prime}_1\,C_1}$, and $I_{2}(B^{\prime},C^{\prime})=CHSH(B^{\prime},C^{\prime})\coloneqq\mean{B^{\prime}_0\,C^{\prime}_0}+\mean{B^{\prime}_0\,C^{\prime}_1}+\mean{B^{\prime}_1\,C^{\prime}_0}-\mean{B^{\prime}_1\,C^{\prime}_1}$; and $\beta_{\rm L}=2$ and $\beta_{\rm NS}=4$. With this choice, Eq. \eqref{sum_generic_app} gives the overall Bell inequality	
\begin{equation}
\label{sum_CHSH_app}
I_{3}=CHSH(A,B)+CHSH(A^{\prime},C)+CHSH(B^{\prime},C^{\prime}) \leq 10.
\end{equation}
Importantly, this inequality can be violated with three Popescu-Rohrlich boxes $p_{\mathrm{PR}}(a,b \vert x,y)\coloneqq (1/2)\delta_{a \oplus b, xy}$  \cite{Popescu1994} distributed among the three parties, such that the overall tripartite correlations are given by $p(a_1,a_1^{\prime},a_2,a_2^{\prime},a_3,a_3^{\prime} \vert x_1,x_2,x_3)= p_{\mathrm{PR}}(a_1,a_2 \vert x_1,x_2)\,p_{\mathrm{PR}}(a_1^{\prime},a_3 \vert x_1,x_3)\,p_{\mathrm{PR}}(a_2^{\prime},a_3^{\prime} \vert x_2,x_3)$. The latter correlations yield the maximal algebraic value $12$ for the lhs of \eqref{sum_CHSH_app}. This shows that the  causal Bell classes  $\{(1),(2),(1,2,3)\}$, $\{(1),(2,3),(1,2,3)\}$ and $\{(1,3),(2,3),(1,2,3)\}$ are nonsignalling interesting.

Furthermore, a very surprising fact arises in the scenario of 3 inputs and 16  outputs (4 bits) per party. There, quantum correlations exist that are incompatible with the three causal classes. More precisely, each output can now be represented by four bits: $A_1$, $A_2$, $A^{\prime}_1$, and $A^{\prime}_2$ for Alice, $B_1$, $B_2$, $B^{\prime}_1$, and $B^{\prime}_2$ for Bob, and $C_1$, $C_2$, $C^{\prime}_1$, and $C^{\prime}_2$ for Charlie. Then, $I_{2}$ is now chosen as a bipartite Bell inequality $I_{\mathrm{PM}}$, for 3 inputs and 4 outputs per party, associated to the so-called Peres-Mermin square \cite{Yang05,Aolita11}. See Eq. (7) in Ref. \cite{Aolita11}, for instance, for an explicit expression of $I_{\mathrm{PM}}(A_1,A_2,B_1, B_2)$. In turn, for this inequality the local and maximal nonsignalling bounds are $\beta_{\rm L}=7$ and $\beta_{\rm NS}=9$, respectively. The interesting feature of $I_{\mathrm{PM}}$ for our purposes is that it can be violated by quantum correlations obtained from a maximally entangled state of two ququarts, or, equivalently, two maximally entangled states of two qubits, up to the algebraic maximal value  $\beta_{\rm NS}=9$.
Thus, with this choice, Eq. \eqref{sum_generic_app} gives the overall Bell inequality
\begin{equation}
\label{sum_PM_app}
I_{3}=I_{\mathrm{PM}}(A_1,A_2,B_1, B_2)+I_{\mathrm{PM}}(A^{\prime}_1,A^{\prime}_2,C_1, C_2)+ I_{\mathrm{PM}}(B^{\prime}_1,B^{\prime}_2,C^{\prime}_1, C^{\prime}_2) \leq 25,
\end{equation}
which can be maximally violated  up to the algebraic maximal value 27 with three maximally entangled states of two ququarts each appropriately distributed among the three parties. Another construction with equivalent implications would be to take $I_{2}$ as the celebrated chained inequality $I_{\mathrm{chain}}$, with 4 outputs (2 bits) and different numbers of inputs per party. This can also be maximally violated up to its algebraic maximal value with quantum correlations. A maximally entangled state of just two qubits (instead of ququarts) is needed for this choice, but at the expenses of requiring an infinitely large number of inputs \cite{BKP06}. Nevertheless, we note that the experimental violations obtained in Refs. \cite{Stuart12} and \cite{Christensen15} for 7 and 18 inputs, respectively, would suffice for a (non-maximal) violation of $I_3$. Either way, we conclude that the causal Bell classes $\{(1),(2),(1,2,3)\}$, $\{(1),(2,3),(1,2,3)\}$ and $\{(1,3),(2,3),(1,2,3)\}$ not only cannot reproduce all quantum correlations but they also satisfy Bell inequalities that are violated by quantum correlations up to the algebraic maximal value.

To end up with, we note that the causal Bell classes $\{(1),(2),(1,2,3)\}$, $\{(1),(2,3),(1,2,3)\}$ and $\{(1,3),(2,3),(1,2,3)\}$ correspond to
 \emph{totally paired} correlations, in the terminology of Ref. \cite{Jones2005}. We emphasise that the question about the existence of quantum correlations more non-local than totally paired models had been open since the work of  Ref. \cite{Jones2005}.

\begin{proof}[Proof that $\{(1),(2),(1,2,3)\}$, $\{(1),(2,3),(1,2,3)\}$ and $\{(1,3),(2,3),(1,2,3)\}$ have LHV marginals over Alice and Bob]
Since the set of all correlations produced by Bayesian networks with respect to $\{(1),(2),(1,2,3)\}$ and $\{(1),(2,3),(1,2,3)\}$ is a subset of the set of all correlations produced by Bayesian networks with respect to $\{(1,3),(2,3),(1,2,3)\}$, it suffices to prove the claim for the latter correlations only.

Consider then arbitrary correlations $\boldsymbol{p}$ produced by a generic Bayesian network with respect to $\{(1,3),(2,3),(1,2,3)\}$, with elements $p(a_1,a_2,a_3,\vert x_1,x_2,x_3)$. Then, the marginal correlations over Alice and Bob will have elements
\begin{eqnarray}
\nonumber
p(a_1,a_2 \vert x_1,x_2)\coloneqq & & \sum_{a_3,x_3,\lambda}\, p(a_1,a_2,a_3,x_3,\lambda \vert x_1,x_2) \\
\label{eq:Bayes_LH}
= & & \sum_{a_3,x_3,\lambda} p(a_1\vert a_2,a_3,x_1, x_2,x_3,\lambda)\,p(a_2 \vert a_3, x_1, x_2,x_3,\lambda)\, p(a_3\vert x_1,x_2,x_3,\lambda)\,p(x_3,\lambda \vert x_1,x_2)\\
\label{eq:Cond_Indep_LH}
= & & \sum_{a_3,x_3,\lambda} p(a_1 \vert x_1,x_3,\lambda)\, p(a_2 \vert x_2,x_3,\lambda)\, p(a_3 \vert x_1,x_2,x_3,\lambda)\,p(x_3,\lambda) \\
\label{eq:Sum_C_LH}
= & & \sum_{x_3,\lambda} p(a_1\vert x_1,x_3,\lambda)\, p(a_2 \vert x_2,x_3,\lambda)\,p(x_3,\lambda) \\
\label{eq:def_lambda_prime}
= & & \sum_{\lambda^{\prime}} p(a_1 \vert x_1,\lambda^{\prime})\,p(a_2 \vert x_2,\lambda^{\prime})\,p(\lambda^{\prime})
\end{eqnarray}
where Bayes rule has been used in Eq. \eqref{eq:Bayes_LH}, Eq. (2) in the main text has been used for Eq. \eqref{eq:Cond_Indep_LH}, Eq. \eqref{eq:Sum_C_LH} follows from summing over $c$, and Eq. \eqref{eq:def_lambda_prime} follows from the identification $\lambda^{\prime}\coloneqq x_3,\lambda $. The right-hand side of Eq. \eqref{eq:def_lambda_prime} manifestly define LHV bipartite correlations for Alice and Bob. Clearly, analogous arguments also hold for all party-exchange symmetries of $\{(1,3),(2,3),(1,2,3)\}$.
\end{proof}

\subsection{Proof that the number of inequivalent causal Bell classes for $N=3$ is 8}
\label{sec:8ineq}

In the end of this subsection, we prove that the causal Bell classes $\{(1),(2),(1,2,3)\}$, $\{(1),(2,3),(1,2,3)\}$ and $\{(1,3),(2,3),(1,2,3)\}$, in the second, third, and fourth levels of the hierarchy, respectively, are equivalent, i.e., all three actually collapse into a same single class, called  the star class. The resulting class is, in addition, inequivalent to the other 7 nonsignalling interesting classes, because, as we show in Sec. \ref{sec:7_quantum}, it can reproduce all nonsignalling correlations for the case of binary inputs and outputs while the other 7 classes cannot.

Now we show that the remaining 7 nonsignalling interesting classes are all inequivalent. To this end, we make extensive use of table \ref{table_noise}. There, for each one of the 46 extremal correlations \cite{Pironio2011} of the nonsignalling polytope for $N=3$ and 2 inputs and 2 outputs, we show the critical noise rates $\nu_{\mathrm{crit}}$ at which the convex combination $\boldsymbol{p}(\nu)=(1-\nu)\,\boldsymbol{p}_{\mathrm{ext}}+\nu\,\boldsymbol{p}_{\mathrm{WN}}$, with $\boldsymbol{p}_{\mathrm{ext}}$ the extremal correlations under scrutiny and $\boldsymbol{p}_{\mathrm{WN}}$ the correlations corresponding to pure white noise (i.e., the uniform distribution), enters the causality polytope of each causal Bell class. That is, $\nu_{\mathrm{crit}}$ is such that $\boldsymbol{p}(\nu)$ is compatible with the causal Bell class in question for all $\nu\geq\nu_{\mathrm{crit}}$. 

\rafael{We notice that to that aim is enough to consider only one representative of each of the 46 classes. This follows from the fact that for each representative we solve a feasibility problem (see Appendix C) that answers the question whether such probability lies inside the associated causal polytope or not. Every other extremal point within the same class (also considering the noise) can be obtained by this representative by symmetry operations (permutation of parties, inputs and outputs and combinations thereof) that can be understood as a simple relabel of variables and therefore does not change the answer to the feasibility problem (that optimizes over all such relabellings).} 

\begin{table}[t]
  \begin{tabular}{|l|l|l|l|l|l|l|l|}
    \hline
    \multirow{2}{*}{Box Number} &
      \multicolumn{1}{c|}{0th-level} &
      \multicolumn{1}{c|}{1st-level} &
      \multicolumn{3}{c|}{2nd-level} &
      \multicolumn{2}{c|}{3rd-level} \\
    & $\left\{ (1),(2),(3) \right\}$ & $\left\{ (1),(1,2),(3) \right\}$ & \{((1,2),(1,2),(3)\} & \{(1),(1,2),(1,3)\}  & \{(1),(1,2),(2,3)\}  & \{(1,2),(1,2),(1,3)\} & \{(1,3),(1,2),(2,3)\} \\ \hline
1	 & 	0	 & 	0	 & 	0	 & 	0	 & 	0	 & 	0	 & 	0	 	 	 \\ \hline
2	 & 	2/3	 & 	0	 & 	0	 & 	0	 & 	0	 & 	0	 & 	0	 	 	 \\ \hline
3	 & 	1/2	 & 	0	 & 	0	 & 	0	 & 	0	 & 	0	 & 	0	 	 	 \\ \hline
4	 & 	2/5	 & 	0	 & 	0	 & 	0	 & 	0	 & 	0	 & 	0	 	 	 \\ \hline
5	 & 	1/2	 & 	0	 & 	0	 & 	0	 & 	0	 & 	0	 & 	0	 	 	 \\ \hline
6	 & 	1/2	 & 	0	 & 	0	 & 	0	 & 	0	 & 	0	 & 	0	 	 	 \\ \hline
7	 & 	1/2	 & 	0	 & 	0	 & 	0	 & 	0	 & 	0	 & 	0	 	 	 \\ \hline
8	 & 	1/2	 & 	0	 & 	0	 & 	0	 & 	0	 & 	0	 & 	0	 	 	 \\ \hline
9	 & 	1/2	 & 	1/6	 & 	0	 & 	1/6	 & 	0	 & 	0	 & 	0	 	 	 \\ \hline
10	 & 	3/5	 & 	1/5	 & 	0	 & 	1/5	 & 	0	 & 	0	 & 	0	 	 	 \\ \hline
11	 & 	1/2	 & 	0	 & 	0	 & 	0	 & 	0	 & 	0	 & 	0	 	 	 \\ \hline
12	 & 	1/2	 & 	0	 & 	0	 & 	0	 & 	0	 & 	0	 & 	0		 	 \\ \hline
13	 & 	1/2	 & 	2/11	 & 	4/37	 & 	2/11	 & 	0	 	 & 	0	 & 	0	  \\ \hline
14	 & 	1/2	 & 	4/19	 & 	1/7	 & 	4/19	 & 	0	 	 & 	0	 & 	0		 \\ \hline
15	 & 	3/5	 & 	3/13	 & 	1/7	 & 	3/13	 & 	1/9	 	 & 	1/9	 & 	0		 \\ \hline
16	 & 	1/2	 & 	2/11	 & 	4/25	 & 	2/11	 & 	0	 	 & 	0	 	 & 	0	 \\ \hline
17	 & 	4/7	 & 	4/19	 & 	4/23	 & 	4/19	 & 	1/7	 	 & 	1/7	 	 & 	0	 \\ \hline
18	 & 	3/5	 & 	5/21	 & 	1/5	 & 	2/9	 & 	1/5	 	 & 	1/5	 & 	0	 	 \\ \hline
19	 & 	8/15	 & 	2/7	 & 	4/19	 & 	2/7	 & 	0	 	 & 	0	 & 	0	 	 \\ \hline
20	 & 	15/29	 & 	16/65	 & 	3/13	 & 	4/19	 	 & 	0	 & 	0		 & 	0	 \\ \hline
21	 & 	1/2	 & 	2/7	 & 	4/17	 & 	2/7	 & 	1/5	 	 & 	1/5	 & 	0	 	 \\ \hline
22	 & 	1/2	 & 	8/29	 & 	1/4	 & 	8/29	 & 	1/4		 & 	1/4	 & 	0		 \\ \hline
23	 & 	1/2	 & 	1/4	 & 	1/4	 & 	4/19	 & 	0	 	 & 	0	 & 	0	 	 \\ \hline
24	 & 	4/7	 & 	1/4	 & 	1/4	 & 	8/35	 & 	0	 	 & 	0	 & 	0	 	 \\ \hline
25	 & 	4/7	 & 	1/4	 & 	1/4	 & 	0	 & 	0	 	 & 	0	 & 	0	 	 \\ \hline
26	 & 	4/7	 & 	8/29	 & 	1/4	 & 	8/29	 & 	1/4	 	 & 	1/4	 & 	0	 	 \\ \hline
27	 & 	1/2	 & 	2/7	 & 	1/4	 & 	2/7	 & 	1/5		 & 	1/5	 & 	1/12	 	 \\ \hline
28	 & 	4/7	 & 	10/37	 & 	1/4	 & 	8/35	 & 	0	 & 	0	 & 	0		 \\ \hline
29	 & 	4/7	 & 	4/13	 & 	10/39	 & 	4/13	 & 	1/4		 & 	1/4	 & 	5/33	 	 \\ \hline
30	 & 	15/29	 & 	2/7	 & 	2/7	 & 	2/7	 & 	2/7	  & 	2/7	 & 	0	 	 \\ \hline
31	 & 	15/29	 & 	2/7	 & 	2/7	 & 	2/7	 & 	2/7	 	 & 	2/7	 & 	0	 	 \\ \hline
32	 & 	15/29	 & 	2/7	 & 	2/7	 & 	8/33	 & 	0	 	 & 	0	 & 	0	 	 \\ \hline
33	 & 	15/29	 & 	2/7	 & 	2/7	 & 	8/33	 & 	0	 	 & 	0	 & 	0	 	 \\ \hline
34	 & 	1/2	 & 	1/3	 & 	1/3	 & 	1/3	 & 	1/3	 	 & 	1/3	 & 	0	 	 \\ \hline
35	 & 	1/2	 & 	1/3	 & 	1/3	 & 	4/19	 & 	0		 & 	0	 & 	0	 	 \\ \hline
36	 & 	1/2	 & 	1/3	 & 	1/3	 & 	1/3	 & 	1/3	 	 & 	1/3	 & 	0	 	 \\ \hline
37	 & 	1/2	 & 	1/3	 & 	1/3	 & 	1/5	 & 	0	 	 & 	0	 & 	0	 	 \\ \hline
38	 & 	4/7	 & 	1/3	 & 	1/3	 & 	4/19	 & 	1/7	 	 & 	1/7	 & 	0	 	 \\ \hline
39	 & 	4/7	 & 	1/3	 & 	1/3	 & 	1/3	 & 	1/3	 	 & 	1/3	 & 	0	 	 \\ \hline
40	 & 	1/2	 & 	1/3	 & 	1/3	 & 	0	 & 	0	 	 & 	0	 & 	0	 	 \\ \hline
41	 & 	1/2	 & 	1/3	 & 	1/3	 & 	1/3	 & 	1/3	 	 & 	1/3	 & 	0	 	 \\ \hline
42	 & 	1/2	 & 	1/3	 & 	1/3	 & 	0	 & 	0	 	 & 	0	 & 	0	 	 \\ \hline
43	 & 	8/15	 & 	4/11	 & 	4/11	 & 	11/54	 & 	1/8	 	 & 	1/8	 & 	2/23	 	 \\ \hline
44	 & 	3/5	 & 	3/8	 & 	3/8	 & 	1/3	 & 	1/3	 	 & 	1/3	 & 	1/4	 	 \\ \hline
45	 & 	1/2	 & 	1/2	 & 	1/2	 & 	0	 & 	0	  & 	0	 & 	0	 	 \\ \hline
46	 & 	1/2	 & 	1/2	 & 	1/2	 & 	1/2	 & 	1/2		 & 	1/2	 & 	0	 	 \\ \hline
      \end{tabular}
\caption{
\label{table_noise}
Table with the critical noise rates $\nu_{\mathrm{crit}}$ at which convex combinations of white noise with one of the extremal correlations of the tripartite scenario with 2 inputs and 2 outputs give correlations inside each causal Bell class. Each row corresponds to one of the 46 nonsignalling extremal correlations of the tripartite scenario with 2 inputs and 2 outputs according to the numbering of Ref. \cite{Pironio2011}. Each column corresponds to one of the seven classes that are nonsignalling interesting in this scenario. The three IO BDAGs $\{(1),(2),(1,2,3)\}$, $\{(1),(2,3),(1,2,3)\}$ and $\{(1,3),(2,3),(1,2,3)\}$ do not appear in the table because their associated causal Bell class is nonsignalling boring for $N=3$ and 2 inputs and 2 outputs, i.e., it displays $\nu_{\mathrm{crit}}=0$ for all 46 extremal correlations.}
\end{table}

Almost all of the 7 classes display different values of $\nu_{\mathrm{crit}}$ for some causal Bell class, which allows one to distinguish them and thus conclude that they are not equivalent. For instance, the causal Bell class $\{(1,3),(1,2),(2,3)\}$ in the third level can reproduce convex combinations with the extremal correlation $41$ in Table \ref{table_noise} even for $v=0$, where for all other classes the reproduction is only possible for $v \geq 1/3$, thus proving a finite gap between $\{(1,3),(1,2),(2,3)\}$ and the remaining 6 classes. A similar argument can be applied to all pairs of classes among the 6 remaining classes (the light-grey shaded ones in Fig. 1 of the main text) except for the pair $\{(1,2),(1,2),(1,3)\}$ and  $\{(1),(1,2),(2,3)\}$.

According to Table \ref{table_noise}, the classes $\{(1,2),(1,2),(1,3)\}$, in the third level, and  $\{(1),(1,2),(2,3)\}$, in the second level, possess the same critical noise rates for all 46 extremal correlations. However, here, one can consider the convex combination $\boldsymbol{p}(\nu)=(1-\nu)\,\boldsymbol{p}_{\mathrm{ext}_1}+\nu\,\boldsymbol{p}_{\mathrm{ext}_{38}}$, where $\boldsymbol{p}_{\mathrm{ext}_1}$ and $\boldsymbol{p}_{\mathrm{ext}_{38}}$ are respectively the first and thirty-eighth extremal correlations. Solving a linear program similar to the one used to generate Table \ref{table_noise}, one sees that $\{(1),(1,2),(2,3)\}$ can reproduce the convex combination only for $\nu=1$, while  $\{(1,2),(1,2),(1,3)\}$ can do it for all $\nu \geq 1/7$. This proves the inequivalence between the two classes and thus concludes the proof.

\begin{proof}[Proof that  $\{(1),(2),(1,2,3)\} \leftrightarrow \{(1),(2,3),(1,2,3)\} \leftrightarrow \{(1,3),(2,3),(1,2,3)\}$]

The $\leftarrow$ implications follows automatically from the hierarchy. So we prove the $\rightarrow$ implications. Consider arbitrary nonsignaling correlations $\boldsymbol{p}$, with elements $p(a_1,a_2,a_3 \vert x_1,x_2,x_3)$. Then, it holds that
\begin{align}
 \label{eq:proof_collapse}
 & p(a_1,a_2,a_3 \vert x_1,x_2,x_3) =  p(a_3 \vert a_1,a_2, x_1,x_2,x_3)\, p(a_1,a_2\vert x_1,x_2,x_3) =  p(a_3 \vert a_1,a_2, x_1,x_2,x_3)\, p(a_1,a_2\vert x_1,x_2) \nonumber \\
 & \forall\  a_1,\ a_2,\ a_3, \ x_1, \ x_2,  x_3,
\end{align}
where the first equality follows from  Bayes rule and the second one from the fact that $\boldsymbol{p}$ is no-signalling.

Assume now that $\boldsymbol{p}$ is produced by a Bayesian network with respect to one of the three IO BDAGs $\{(1),(2),(1,2,3)\}$, $\{(1),(2,3),(1,2,3)\}$, or $\{(1,3),(2,3),(1,2,3)\}$. Then, the marginal distribution $p(a_1,a_2\vert x_1,x_2)$ for Alice and Bob, the second factor in the RHS of Eq. \eqref{eq:proof_collapse}, is restricted to the same set of (LHV bipartite) correlations, regardless of which one out of the three BDAGs does indeed generate $\boldsymbol{p}$, as proven in the end of Sec. \ref{sec:7_quantum}. On the other hand, note, in addition, that, for all three IO BDAGs, the term $p(a_3 \vert a_1,a_2, x_1,x_2,x_3)$, the first factor in the RHS of Eq. \eqref{eq:proof_collapse}, has no restriction whatsoever and spans the whole set of conditional probability distributions of $c$ given $a_1$, $a_2$, $x_1$, $x_2$, and $x_3$. This follows from the fact that, for all three BDAGs, Charlie has access to Alice and Bob`s inputs as well as to the hidden variable, so he can reproduce any arbitrary conditional distribution of the form $p(a_3 \vert a_1,a_2, x_1,x_2,x_3)$. In other words, for neither of the two factors in the RHS of Eq. \eqref{eq:proof_collapse} do the  arrows from Charlie to Alice or Bob provide any extra capability at reproducing no-signaling correlations. This implies that $\boldsymbol{p}$ itself is restricted to the same set of correlations for all three BDAGs.

The same argument holds for all the other party-exchange symmetries of $\{(1),(2),(1,2,3)\}$, $\{(1),(2,3),(1,2,3)\}$, or $\{(1,3),(2,3),(1,2,3)\}$, and, clearly, also for convex combinations of correlations produced by them. This proves the $\rightarrow$ implications.
\end{proof}

As a final comment, we note that the last proof generalises straightforwardly  to the case of arbitrary $N$. In other words, all causal Bell classes represented by IO BDAGs containing a star, i.e., for which all inputs go to the output of one party while any other locality relaxation involves the latter party, are equivalent:
\begin{align}
\nonumber
\{(1),(2), \hdots,(N-1), (1,2, \hdots ,N-1,N)\}&\leftrightarrow&\{(1),(2), \hdots, (N-2), (N-1,N), (1,2, \hdots ,N-1,N)\} \\
\nonumber
&\leftrightarrow&\\
\nonumber
&\ \ \vdots&\\
&\leftrightarrow&\{(1,N),(2,N), \hdots, (N-1,N), (1,2, \hdots ,N-1,N)\}.
\end{align}

\section{Number of causal Bell classes in the fourpartite case}
\label{sec:Four-partite}
In this section, as a further example of the applicability of our machinery, we list all the IO BDAGs, excluding party-exchange redundancies, that appear for $N=4$. The complete hierarchy possesses a total of 52 causal Bell classes, taking into account both nonsignaling interesting as well as nonsignaling boring ones, and disregarding collapses among different classes.

\begin{table}[h]
  \begin{tabular}{|c|c|}
    \hline
Level &  Number of causal Bell classes  \\ \hline
0	 & 	1 		 	 	 \\ \hline
1	 & 	1			 	 	 \\ \hline
2	 & 	5			 	 	 \\ \hline
3	 & 	13			 	 	 \\ \hline
4	 & 	27			 	 	 \\ \hline
5	 & 	38			 	 	 \\ \hline
6	 & 	48			 	 	 \\ \hline
7	 & 	38			 	 	 \\ \hline
8	 & 	27			 	 	 \\ \hline
9	 & 	13			 	 	 \\ \hline
10	 & 	5			 	 	 \\ \hline
11	 & 	1			 	 	 \\ \hline
12	 & 	1			 	 	 \\ \hline
\end{tabular}
\caption{Table with all the IO BDAGs that arise in the fourpartite scenario excluding party-exchange redundancies. The complete hierarchy possesses 52 causal Bell classes, including both nonsignalling boring and interesting ones, distributed in 12 levels.}
\label{table_4}
\end{table}
\end{widetext}
\end{document}